\documentclass[journal]{IEEEtran}
\usepackage{cite}
\ifCLASSINFOpdf
   \usepackage[pdftex]{graphicx}
\else
  \usepackage[dvips]{graphicx}
\fi
\usepackage{amsmath,amssymb,amsthm,mathrsfs,amsfonts,dsfont,stmaryrd}
\usepackage{array}
\usepackage{dsfont}
\usepackage{amsbsy}
\usepackage{epstopdf}
\usepackage{graphicx,color}

\newtheorem{lemma}{Lemma}
\newtheorem{proposition}{Proposition}
\newtheorem{remark}{Remark}
\newtheorem{corollary}{Corollary}
\long\def\symbolfootnote[#1]#2{\begingroup%
\def\thefootnote{\fnsymbol{footnote}}\footnote[#1]{#2}\endgroup}
\newtheorem{theorem}{Theorem}
\newtheorem{definition}{Definition}

\newcommand{\dv}{\mathbf} % determenistic vector
\newcommand{\mc}{\mathcal} % determenistic vector
\newcommand{\mkv}{-\!\!\!\!\minuso\!\!\!\!-}

\newcommand\blfootnote[1]{%
  \begingroup
  \renewcommand\thefootnote{}\footnote{#1}%
  \addtocounter{footnote}{-1}%
  \endgroup
}

 \allowdisplaybreaks

\usepackage{algorithm} 
\usepackage{algpseudocode}

\algnewcommand{\Inputs}[1]{%
  \State \textbf{Inputs:}
  \Statex \hspace*{\algorithmicindent}\parbox[t]{.8\linewidth}{\raggedright #1}
}
\algnewcommand{\Initialize}[1]{%
  \State \textbf{initialization}
  \Statex \hspace*{\algorithmicindent}\parbox[t]{.95\linewidth}{\raggedright #1}
}

\begin{document}
\fontencoding{OT1}\fontsize{10}{11}\selectfont

\title{On the Capacity of Cloud Radio Access Networks with Oblivious Relaying }
\author{I\~naki Estella Aguerri$^{\dagger}$ \qquad \quad Abdellatif Zaidi$^{\dagger}$ $^{\ddagger}$   \quad \quad Giuseppe Caire$^{\star}$  \quad \quad Shlomo Shamai (Shitz)$^{\ast}$ \vspace{0.3cm}\\
% $^{\dagger}$ Mathematics and Algorithmic Sciences Lab.\\
\small{$^{\dagger}$France Research Center, Huawei Technologies, Boulogne-Billancourt, 92100, France\\
$^{\ddagger}$ Universit\'e Paris-Est, Champs-sur-Marne, 77454, France\\
$\star$ Technische Universit¨at Berlin, 10587 Berlin, Germany \\
$\ast$ Technion Institute of Technology, Technion City, Haifa 32000, Israel\\
{\{\tt inaki.estella@huawei.com, abdellatif.zaidi@u-pem.fr, caire@tu-berlin.de, sshlomo@ee.technion.ac.il\}\vspace{-5mm}}}
}

% make the title area
\maketitle
\begin{abstract} 
We study the transmission over a network in which users send information to a remote destination through relay nodes that are connected to the destination via finite-capacity error-free links, i.e., a cloud radio access network. The relays are constrained to operate without knowledge of the users' codebooks, i.e., they perform oblivious processing -- The destination, or central processor, however, is informed about the users' codebooks. We establish a single-letter characterization of the capacity region of this model for a class of discrete memoryless channels in which the outputs at the relay nodes are independent given the users' inputs. We show that both relaying \`a-la Cover-El Gamal, i.e.,  compress-and-forward with joint decompression and decoding, and ``noisy network coding'', are optimal. The proof of the converse part establishes, and utilizes,  connections with the Chief Executive Officer (CEO) source coding problem under logarithmic loss distortion measure. 
Extensions to general discrete memoryless channels are also investigated. In this case, we establish inner and outer bounds on the capacity region. For memoryless Gaussian channels within the studied class of channels, we characterize the capacity under Gaussian\mbox{ channel inputs.} \blfootnote{The work of G. Caire is supported by an Alexander von Humboldt Professorship.
The work of S.Shamai has been supported by the European Union's Horizon 2020 Research And Innovation Programme, grant agreement no. 694630.} 
\end{abstract}

% The paper headers
%\markboth{IEEE Transactions on Communications}%
%{Submitted paper}

% keywords
%\sqrt{•}
\IEEEpeerreviewmaketitle

%%%%%%%%%%%%%%%%%
\vspace{-5mm}

\section{Introduction}

Cloud radio access networks (CRAN) provide a new architecture for next-generation wireless cellular systems in which base stations (BSs) are connected to a cloud-computing central processor (CP) via error-free finite-rate fronthaul links. This architecture is generally seen as an efficient means to increase spectral efficiency in cellular networks by enabling joint processing of the signals received by multiple BSs at the CP, thus alleviating the effect of interference. Other advantages include low cost deployment and flexible network utilization~\cite{Peng2015FHConstCRAN}.  

In a CRAN network, each BS acts essentially as a relay node; and so can in principle implement any relaying strategy, e.g., decode-and-forward~\cite[Theorem 1]{Cover:1979}, compress-and-forward~\cite[Theorem 6]{Cover:1979} or combinations of them. Relaying strategies in CRANs can be divided roughly into two classes: i) strategies that require the relay nodes to know the users' codebooks (i.e., modulation, coding),  such as decode-and-forward, compute-and-forward~\cite{Nazer:IT:2011, HongCaire:IT:2013} or variants of it, and ii) strategies in which the relay nodes operate without knowledge of the users' codebooks, often referred to as \textit{oblivious relay processing} (or \textit{nomadic transmitter})~\cite{Sanderovich2008:IT:ComViaDesc, Simeone2011:IT:CodebookInfoOutBand,Sanderovich:2009:IT}. This second class is composed essentially of strategies in which the relays implement forms of compress-and-forward~\cite{Cover:1979}, such as successive Wyner-Ziv compression~\cite{Hwan:2013:VT,ZhouYu:2013:JSAC, Park:2013:SPLett} and noisy-network coding \cite{Lim:IT:2011:NoisyNetwork}. Schemes combining the two apporaches have been shown to possibly outperform the best of the two in \cite{Estella2016CQCF}, especially in scenarios in which there are more users than relay nodes. 

In the spirit, however, a CRAN architecture is usually envisioned as one in which BSs operate as simple radio units (RUs) that are constrained to implement only the radio functionalities, such as analog-to-digital conversion and filtering, while the baseband functionalities are migrated to the CP. For this reason, while relaying schemes that involve partial or full decoding of the users' codewords can sometimes offer rate gains, they do not seem to be suitable in practice --  In fact, such schemes assume that all or a subset of the relay nodes are fully aware (at all times) of the codebooks and encoding used by the users; and the signaling required to convey such information is generally prohibitive, particularly as networks become large. Instead, schemes in which relays perform oblivious processing are preferred. Oblivious processing was first introduced in \cite{Sanderovich2008:IT:ComViaDesc}. The basic idea is that of using \textit{randomized encoding} to model lack of information about codebooks.  For related works, the reader may refer to~\cite{Simeone2011:IT:CodebookInfoOutBand, Dytso2014, Tian2011, Simeone2011:IT:RobustCommDetProcc}. In particular, ~\cite{Simeone2011:IT:RobustCommDetProcc} extends the original definition of oblivious processing of \cite{Sanderovich2008:IT:ComViaDesc}, which rules out time-sharing, to include settings in which encoders are allowed to switch among different codebooks and oblivious nodes are unaware of the codebooks but are given, or can acquire, time- or frequency-schedule information, which is generally less difficult to obtain. The framework is termed therein as ``\textit{oblivious processing with enabled time-sharing}".  

In this work, we consider transmission over a CRAN in which the relay nodes are constrained to operate without knowledge of the users' codebooks, i.e., are oblivious, and only know time- or frequency-sharing information. Focusing on a class of discrete memoryless channels in which the relay outputs are independent conditionally on the users' inputs, we establish a single-letter characterization of the capacity region of this class of channels. We show that relaying \`a-la Cover-El Gamal, i.e.,  compress-and-forward with joint decompression and decoding, or noisy network coding, are optimal. For the proof of the converse part, we utilize useful connections with the Chief Executive Officer (CEO) source coding problem under logarithmic loss distortion measure~\cite{Courtade2014LogLoss}. For memoryless Gaussian channels within this class, we characterize the capacity under Gaussian channel inputs. Extensions to general discrete memoryless channels are also investigated. In this case, we establish inner and outer bounds on the capacity region.

%\subsection*{Notations}
\textit{Notation:} Throughout, we use the following notation. Lower case letters denote scalars, e.g., $x$; upper case letters denote random variables, e.g., $X$, boldface lower case letters denote vectors, e.g., $\dv x$, and boldface upper case letters denote matrices, e.g., $\dv X$. Calligraphic letters denote sets, e.g., $\mathcal{X}$; and the cardinality of set $\mc X$ is denoted by $|\mc X|$. For a set of integers $\mc K$, the notation $X_{\mc K}$ denotes the set of random variables $\{X_k\}$ with indices $k$ in the set $\mc K$, i.e., $X_{\mc K}=\{X_k\}_{k \in \mc K}$.

\section{System Model}\label{sec:System}

Consider the discrete memoryless CRAN model shown in Figure~\ref{fig:Schm}. In this model, a set of users communicate with a central processor (CP) through a set of relay nodes that are connected to the CP via error-free finite-rate fronthaul links. Let $\mc L=\{1,\hdots,L\}$ denote the set of users, and $\mc K=\{1,\hdots,K\}$ denote the set of relays, and let $C_k$ be the capacity of the link connecting relay node $k$ to the CP, $k \in \mc K$. 

\noindent Similar to~\cite{Simeone2011:IT:CodebookInfoOutBand}, the relays nodes are constrained to operate without knowledge of the users' codebook and only know time-sharing information, i.e., oblivious relay processing with enabled time sharing. The obliviousness of the relay nodes to the actual codebooks is modeled by the transmitters picking at random their selected codebooks and the relays not aware of the actual codebooks indices. Specifically, the codeword $X^n(F_l, M_l, Q^n)$ transmitted by encoder $l$, $l \in \mc L$, depends not only on the message $M_l \in [1,2^{nR_l}]$, but also on the index $F_l$ which runs over all possible codebooks of the given rate $R_l$, i.e., $F_l \in [1,|\mc X_l|^{n2^{nR_l}}]$ and the time sharing sequence $Q^n$. Formally, the model is defined as follows.  
\begin{enumerate}
\item \textit{Codebooks:} 
Transmitter $l$, $l\in \mathcal{L}$, sends message $m_l\in [1,2^{nR_l}]$ to the CP using a codebook from a set $\{\mathcal{C}_l(F_l)\}$ that is indexed by $F_l\in [1,|\mathcal{X}_l|^{n2^{nR_l}}]$. The index $F_l$ is picked at random and shared with CP, but not to the relays.
\item \textit{Time-sharing sequence:} All terminals, including the relay nodes, are aware of a time-sharing sequence $Q^n$, distributed as $p_{Q^n}(q^n)=\prod_{i=1}^n p_{Q}(q_i)$ for a pmf $p_{Q}(q)$. 

\item  \textit{Encoding functions:} The encoding function at user $l$, $l \in \mc L$, is defined by a pair $(p_{X_l},\phi_l)$ where $p_{X_l}$ is a single-letter pmf and $\phi_l$ is a mapping $\phi_l:[1,|X_l|^{n2^{nR_l}}]\times [1,2^{nR_l}]\times \mathcal{Q}^n\rightarrow \mathcal{X}_l^n$, that maps the given codebook index $F_l$, message $m_l$ and time-sharing variable $q^n$ to a channel input $x_{l}^n=\phi_l(F_l,m_l, q^n)$. The probability of selecting codebook $F_l$, $F_l\in [1,|\mathcal{X}_l|^{n2^{nR_l}}]$, is given by 
\begin{align}
p_{F}(f,q^n) = \prod_{m \in [1,2^{nR}]} p_{X^n|Q^n}(\phi_l(f,m)|q^n), 
\end{align}
where $p_{X^n|Q^n}(x^n|q^n) = \prod_{i=1}^n p_{X|Q}(x_i|q_i)$ for some given pmf $p_{X|Q}(x|q)$.

\item \textit{Relaying functions:} Relay node $k$ , $k \in \mathcal{K}$, is unaware of the codebook indices $F_{\mathcal{L}}= (F_1,\ldots, F_L)$, and maps the received sequence $y_k^n\in \mathcal{Y}_k$ into an index $J_k\in[1,2^{nC_k}]$ as $J_k=\phi_k^r(y_k^n,q^n)$, that it then sends to the CP over the error-free link of capacity $C_k$. The index $J_k$ is then sent the to the CP over the link of capacity $C_k$. 

\item \textit{Decoding function:} Upon receiving the indices $J_{\mathcal{K}} = (J_1\ldots,J_K)$, the CP estimates the users' messages as
\begin{align}
(\hat{m}_1,\ldots,\hat{m}_L)= g(F_1,\ldots,F_L,J_1,\ldots,J_K,q^n),
\end{align}
where $g:[1,|\mathcal{X}_1|^{n2^{nR_1}}]\times\cdots\times [1, |\mathcal{X}_L|^{n2^{nR_L}}]\times [1, 2^{nC_1}]\times \cdots \times [1, 2^{nC_L}]\times \mathcal{Q}^n\rightarrow [1, 2^{nR_1}] \times \hdots \times [1, 2^{nR_L}]$ is the decoder mapping. 
\end{enumerate}

\begin{definition}
A $(n,R_1,\ldots,R_L)$ code for the studied model with oblivious relay processing and enabled time-sharing consist of $L$ encoding functions
%\begin{align}
$\phi_l:[1,|X_l|^{n2^{nR_l}}]\times [1,2^{nR_l}]\times \mathcal{Q}^n\rightarrow \mathcal{X}^n_l$,
%\end{align}
$K$ relaying functions 
$\phi_k^r: \mathcal{Y}^n_k\times \mathcal{Q}^n\rightarrow [1,2^{nC_k}]$,
and a decoding function 
%\begin{align}
$g:[1, |\mathcal{X}_1|^{n2^{nR_1}}]\times\cdots\times [1, |\mathcal{X}_L|^{n2^{nR_L}}]\times [1, 2^{nC_1}]\times\cdots\times [1, 2^{nC_L}] \times \mathcal{Q}^n \rightarrow [1, 2^{nR_1}] \times \hdots [1, 2^{nR_L}]$.
%\end{align}
\end{definition}

\begin{definition}
A rate tuple $(R_{1},\ldots,R_{L})$ is said to be achievable if, for any $\epsilon>0$, there exist a sequence of codes, such that, for sufficiently long blocklength $n$, each user's message can be decoded by the CP at rate at least $R_k$ with vanishing probability of error, i.e.,
\begin{align}
 \mathrm{Pr}\{(M_1,\hdots, M_L)\neq (\hat{M}_1,\hdots, \hat{M}_L)\}\leq \epsilon. 
 \end{align}
\end{definition}

\noindent For given $\mathcal{C}_{\mathcal{K}}$, the capacity region $\mathcal{R}(C_{\mathcal{K}})$ is the closure of all achievable rate tuples $(R_{1},\ldots,R_{L})$.
%----------------------------------------
\begin{figure}[t!]
\centering
\includegraphics[width=0.495\textwidth]{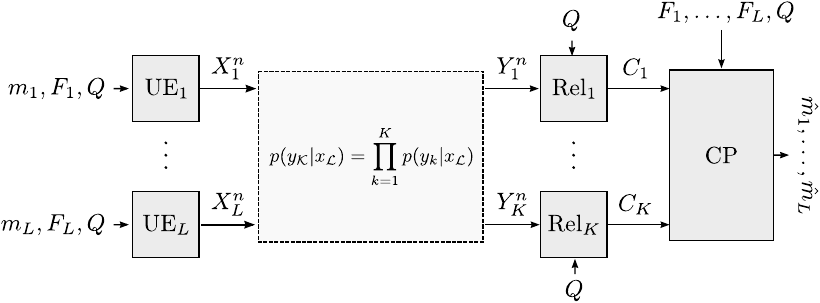}
\vspace{-5mm}
\caption{CRAN model with oblivious relaying and time-sharing.} 
\label{fig:Schm}
\end{figure}

\vspace{0.2cm}

\noindent Due to space limitations, some of the results of this paper are only outlined or given without proofs. The detailed proofs can be found in~\cite{Estella2017OCRAN_Proofs}.

%----------------------------------------
\subsection{Class of Discrete Memoryless Channels}

In this work, we establish the capacity region of the following class of discrete memoryless CRAN channels with oblivious relay processing and enabled time-sharing. In this class, the channel outputs at the relay nodes are independent conditionally on the users' inputs. That is,  
\begin{equation}
p(y_{\mathcal{K}}^n|x_{\mathcal{L}}^n)=\prod_{i=1}^n\prod_{k=1}^Kp(y_{k,i}|x_{\mathcal{L},i}), 
\end{equation}
or, equivalently, the following Markov chain holds,
\begin{align}
Y_{k,i} \mkv X_{\mathcal{L},i} \mkv Y_{\mathcal{K}/k,i}\quad \text{for } k \in \mc K\: \text{and} \: i \in [1,n]. 
\label{eq:MKChain_pmf}
\end{align}
% This class of channels includes most CRAN channels considered in the literature, in particular, the memoryless Gaussian MIMO channel discussed in Section \ref{ssec:Gauss}, e.g., \cite{Hwan:2013:VT,ZhouYu:2013:JSAC, Park:2013:SPLett, Lim:IT:2011:NoisyNetwork, Somekh:2007:IT,DelCoso:2009:TWir, Sanderovich:2009:IT,DBLP:journals/corr/ZhouX0C16}.  

\subsection{Oblivious Relaying with Enabled Time-Sharing}
Similar to~\cite{Simeone2011:IT:CodebookInfoOutBand}, the above constraint of oblivious relay processing with enabled time-sharing means that, in the absence of information regarding the indices $F_{\mathcal{L}}$ and the messages $M_{\mathcal{L}}$, a codeword $x_l^n(F_l,m_l|q^n)$ taken from a $(n,R_l)$ codebook has independent but non-identically distributed entries. 
 \vspace{0.2cm}
\begin{lemma}\label{lem:IIDinput}
Without the knowledge of the selected codebooks indices $(F_1,\ldots,F_L)$, the distribution of the transmitted codewords conditioned on the time-sharing sequence are given by $p_{X_l^n|Q^n}(x_l^n|q^n)=\prod_{i=1}^np_{X_l|Q}(x_{l,i}|q_i)$. Consequently, the channel output at relay $k$, $Y_k^n$  is distributed as
\begin{align}
&p_{Y_k^n|Q^n}(y^n_k|q^n)\nonumber
=\prod_{i=1}^n\sum_{x_1,\ldots, x_L}p_{Y_k|X_{\mathcal{L}}}({y}_{k,i}|x_{\mathcal{L}})\prod_{i=1}^Lp_{X_l|Q}(x_l|q_i).\nonumber
\end{align}
\end{lemma}

\section{Main Results}\label{sec:Main}

\subsection{Capacity Region of Studied Class of CRAN Channels}

The main result of this paper is a single-letter characterization of the capacity region of the class of channels with oblivious relaying and enabled time-sharing that satisfy~\eqref{eq:MKChain_pmf}. The following theorem states the result.

\vspace{0.1cm}
\begin{theorem}\label{th:MK_C_Main}
For the class of discrete memoryless channels given by~\eqref{eq:MKChain_pmf} with oblivious relay processing and enabled time-sharing, a rate tuple $(R_1,\ldots, R_L)$ is achievable if and only if for all $\mathcal{T} \subseteq \mathcal{L}$ and for all $\mathcal{S} \subseteq \mathcal{K}$, we have
\begin{align}
\sum_{t\in \mathcal{T}}R_t\leq& \sum_{s\in \mathcal{S}} [C_s-I(Y_{s};U_{s}|X_{\mathcal{L}},Q)]\nonumber  \\
&+ I(X_{\mathcal{T}};U_{\mathcal{S}^c}|X_{\mathcal{T}^c},Q),\label{eq:MK_C_Main}
\end{align}
for some joint measure of the form $p(q)\prod_{l=1}^L p(x_l|q)\prod_{k=1}^Kp(y_k|x_{\mathcal{L}})\prod_{k=1}^{K}p(u_k|y_k,q)$.
\end{theorem} 

\vspace{0.1cm}

\noindent\textbf{Proof}: The proof of converse part of Theorem~\ref{th:MK_C_Main} is relegated to Section~\ref{sec:Converse}.  The proof of the direct part can be obtained by applying the noisy network coding (NNC) scheme of~\cite[Theorem 1]{Lim:IT:2011:NoisyNetwork}. Alternatively, the rate region of Theorem~\ref{th:MK_C_Main} can also be achieved by a scheme that generalizes that of~\cite[Theorem 1]{Sanderovich:2009:IT}, which is established in the case of a single transmit node, to the case of multiple users and accommodate time-sharing. By  opposition to the NNC scheme, the latter scheme is based on compress-and-forward \`a la Cover-El Gamal with joint decoding and decompression at the CP (CoF-JD).
\qed

\begin{remark}
Key element for the proof of Theorem~\ref{th:MK_C_Main} is the connection with the chief executive officer (CEO) problem. For the case of $m$ encoders, $m \geq 3$, while characterization of the optimal rate-distortion region of this problem for general distortion measures has eluded the information theory, a characterization of the optimal region in the case of logarithmic loss distortion measure has been provided recently in~\cite{ Courtade2014LogLoss}. 
\end{remark}

\begin{remark}
The sum-rate of Theorem~\ref{th:MK_C_Main} can also be achieved by a scheme in which the CP decodes explicitly the compression indices first, and then decodes the users' transmitted messages, i.e., decompression and decoding is not performed jointly. A similar observation is found in~\cite[Theorem 2]{DBLP:journals/corr/ZhouX0C16}. 
\end{remark}

\subsection{Memoryless Gaussian Model}\label{ssec:Gauss}

In this section, we consider a memoryless Gaussian MIMO model of the studied CRAN with oblivious relay processing and enabled time sharing. The channel output at relay node $k$, equipped with $M_k$ receive antennas, is given by 
\begin{equation}
\mathbf{Y}_k = \mathbf{H}_{k,\mathcal{L}}\mathbf{X}+\mathbf{N}_k,
\label{mimo-gaussian-model}
\end{equation}
where $\mathbf{X}\triangleq[\mathbf{X}_1^T,\ldots,\mathbf{X}_L^T]^T$, $\mathbf{X}_{l}\in \mathds{C}^{N_l}$ is the channel input vector of user $l$, $l \in \mc L$, $\mathbf{H}_{k,\mathcal{L}} \triangleq [\mathbf{H}_{k,1},\ldots, \mathbf{H}_{k,L}]$, $\mathbf{H}_{k,l}\in \mathds{C}^{M_k\times N_l}$ is the channel matrix connecting  user $l$ to relay node $k$, and $\mathbf{N}_k\in\mathds{C}^{M_k}$ is the noise vector at relay node $k$, assumed to be Gaussian with $\mathbf{N}_k\sim\mathcal{CN}(\mathbf{0},\mathbf{\Sigma}_k)$. The transmission is subjected to power constraint $\text{Tr}(\mathbf{K}_{l}) \leq P_k$, where $\mathbf{K}_{l}= \mathrm{E}[\mathbf{X}_l\mathbf{X}^H_{l}]$ is the covariance matrix of $\mathbf{X}_{l}$. The noises at the relay nodes are assumed to be independent; and so the studied Gaussian model satisfies the Markov chain~\eqref{eq:MKChain_pmf}.

The result of Theorem~\ref{th:MK_C_Main} can be extended to continuous channels using standard techniques; and so it characterizes the capacity region of the model~\eqref{mimo-gaussian-model}. The computation of this region, however, is not easy as it requires  finding the optimal choices of the involved auxiliary random variables $U_1, \hdots, U_K$. The following theorem characterizes more explicitly the capacity region when the users are constrained to employ Gaussian signaling,
 i.e., for $Q=q$, $\mathbf{X}_{l,q}\sim\mathcal{CN}(\mathbf{0},\mathbf{K}_{l,q})$, for all $l \in \mc L$. 
\vspace{0.1cm}

\begin{theorem}
\label{th:GaussSumCap}
If the input vectors are Gaussian, the capacity region of the memoryless Gaussian MIMO model~\eqref{mimo-gaussian-model} is given by the set of all rate tuples $(R_1,\ldots,R_L)$ satisfying that for all $ \mathcal{T} \subseteq \mathcal{L}$ and all $\mathcal{S} \subseteq \mathcal{K}$
\begin{align}
\sum_{t\in\mathcal{T}}R_{t} \leq&
 \sum_{k\in \mathcal{S}}\left[C_k-\log\frac{|\mathbf{\Sigma}_k|}{|\mathbf{\Sigma}^{-1}_k-\mathbf{B}_k|}\right] \\
 &+ \log \frac{
|\sum_{k\in\mathcal{S}^{c}}\mathbf{H}_{k,\mathcal{T}}^{H}
\mathbf{B}_{k}
\mathbf{H}_{k,\mathcal{T}}+\mathbf{K}^{-1}_{\mathcal{T}}|
}{
|\mathbf{K}_{\mathcal{T}}^{-1}|,
},\label{eq:GaussSumCap}
\end{align}
for some $\mathbf{0}\preceq \mathbf{B}_k\preceq \mathbf{\Sigma}_{k}^{-1}$, where $\mathbf{H}_{k,\mathcal{T}}$ denotes the channel matrix connecting the input $\mathbf{X}_{\mathcal{T}}$ to the output $\mathbf{Y}_k$, formed by concatenating the matrices $\mathbf{H}_{k,l}$, $l\in \mathcal{T}$, horizontally.
\end{theorem}

\vspace{0.1cm}

\begin{remark}
Theorem~\ref{th:GaussSumCap} extends the result of~\cite[Theorem 5]{Sanderovich2008:IT:ComViaDesc} to the case of $L$ users and enabled time-sharing. In addition to showing that under the constraint of Gaussian input signaling, the quantization codewords can be chosen optimally to be Gaussian, the result of Theorem~\ref{th:GaussSumCap} also means that time-sharing is not needed in the memoryless Gaussian case. Recall that, as shown through an example in~\cite{Sanderovich2008:IT:ComViaDesc}, if the relays are aware of the users' codebooks restricting to Gaussian input signaling can be a severe constraint and is generally suboptimal. 
\end{remark}

\begin{remark}
In~\cite{DBLP:journals/corr/ZhouX0C16}, the authors study the questions of optimal fronthaul compression and decoding strategies for uplink CRAN networks without oblivious processing constraints. It is shown that NNC with Gaussian input and Gaussian quantization achieve to within a constant gap of the capacity region of the Gaussian MIMO uplink CRAN. In this paper, we show that if only oblivious relay processing is allowed, NNC and CoF-JD is in fact optimal from a capacity viewpoint. 
\end{remark}
\vspace{-3mm}

\section{General Discrete Memoryless Model}\label{sec:GenOCRAN}

In this section, we focus on general discrete memoryless CRAN channels with oblivious relay processing and time sharing, i.e., the channel outputs at the relays are arbitrarily correlated and the Markov chain~\eqref{eq:MKChain_pmf} does not necessarily hold. We establish bounds on the capacity region of the model. The results extend those of~\cite{Sanderovich2008:IT:ComViaDesc}, which only consider a single transmitter and no time-sharing, to the case of multiple transmitters and allowed time-sharing.  

\noindent The following theorem provides an inner bound on the capacity region of the general DM CRAN model with oblivious relay processing and time sharing.

\begin{theorem}\label{th:NNC_all_MK_inner}
For general DM CRAN channels with oblivious relay processing and enabled time-sharing, the set of rates $(R_1,\ldots,R_L)$ such that for all $\mathcal{T} \subseteq \mathcal{L}$ and all $ \mathcal{S} \subseteq \mathcal{K}$,  
\begin{align}\label{eq:NNC_all_MK_inner}
\sum_{t\in \mathcal{T}}R_t\leq& \sum_{s\in \mathcal{S}} C_s-I(Y_{S};U_{\mathcal{S}}|X_{\mathcal{L}},U_{\mathcal{S}^c},Q) \\ 
&+ I(X_{\mathcal{T}};U_{\mathcal{S}^c}|X_{\mathcal{T}^c},Q),
\end{align}
for some joint measure $p(q)\prod_{l=1}^{L}p(x_l|q) p(y_{\mathcal{K}}|x_{\mathcal{L}})$ $\cdot\prod_{k=1}^{K}p(u_k|y_k,q)$, is achievable.
\end{theorem} 

\noindent We now provide an outer bound on the capacity region of the general DM CRAN model with oblivious relay processing and time-sharing. The following theorem states the result.

\begin{theorem}\label{th:NNC_all_MK_outer}
For general DM CRAN channels with oblivious relay processing and enabled time-sharing, 
if a rate-tuple $(R_1,\ldots,R_L)$ is achievable then for all $\mathcal{T} \subseteq \mathcal{L}$ and all $ \mathcal{S} \subseteq \mathcal{K}$, 
\begin{align}\label{eq:NNC_all_MK_outer}
\sum_{t\in \mathcal{T}}R_t\leq &\sum_{s\in \mathcal{S}} C_s-I(Y_{S};U_{\mathcal{S}}|X_{\mathcal{L}},U_{\mathcal{S}^c},Q) 
\\& + I(X_{\mathcal{T}};U_{\mathcal{S}^c}|X_{\mathcal{T}^c},Q),
\end{align}
for some $(Q,X_{\mathcal{L}},Y_{\mathcal{K}},U_{\mathcal{K}},W)$ distributed according to
$p(q)\prod_{l=1}^{L}p(x_l|q) ~ p(y_{\mathcal{K}}|x_{\mathcal{L}})~p(w|q)$, $u_{k} = f_k(w,y_k,q)$ for $k=[1, K]$, 
for some random variable $W$ and some deterministic functions $\{f_{k}\}$, $k\in \mc K$.
\end{theorem}
%\noindent\textbf{Proof:} The proof of Theorem~\ref{th:NNC_all_MK_outer} appears in Appendix~\ref{app:NNC_all_MK_outer}. \qed
\begin{remark}
The inner bound of Theorem~\ref{th:NNC_all_MK_inner} and the outer bound of Theorem \ref{th:NNC_all_MK_outer} do not coincide in general. This is due to the fact that in Theorem \ref{th:NNC_all_MK_inner}, $U_1,\ldots, U_K$ satisfy the Markov chain $U_k \mkv (Y_k,Q) \mkv (X_{\mathcal{L}}, Y_{\mathcal{L}/k},U_{\mathcal{K}/k})$,  whereas in Theorem \ref{th:NNC_all_MK_outer} each $U_k$ is a function of $Y_k$ but also of a ``common'' random variable $W$. In particular,  the Markov chain $U_k \mkv (Y_k,Q) \mkv U_{\mathcal{K}/k}$ does not necessarily hold for the auxiliary random variables of the outer bound.
\end{remark}

\begin{remark}
As we already mentioned, the class~\eqref{eq:MKChain_pmf} of DM CRAN channels connects with the CEO problem under logarithmic loss distortion measure. The rate-distortion region of this problem is characterized in the excellent contribution~\cite{Courtade2014LogLoss} for an arbitrary number of (source) encoders (see Theorem 3 therein). For general DM CRAN channels, i.e., \textit{without} the Markov chain~\eqref{eq:MKChain_pmf} the model connects with the distributed source coding problem under logarithmic loss distortion measure. While a solution of the latter problem for the case of two encoders has been found in~\cite[Theorem 6]{Courtade2014LogLoss}, generalizing the result to the case of arbitrary number of encoders poses a significant challenge. In fact, as also mentioned in~\cite{Courtade2014LogLoss},  the Berger-Tung inner bound is known to be generally suboptimal (e.g., see the Korner-Marton lossless modulo-sum problem~\cite{Korner:IT:1979HowToEncode}). Characterizing the capacity region of the general DM CRAN model under the constraint of oblivious relay processing and enabled time-sharing poses a similar challenge, except perhaps for the case of two relay nodes, results on which will be reported elsewhere. 
\end{remark}

\section{Proof of Converse Part of Theorem~\ref{th:MK_C_Main}}\label{sec:Converse}

Assume the rate tuple $(R_{1},\ldots,R_L)$ is achievable. Let $\mathcal{T}$ be a set of $\mathcal{L}$, $\mathcal{S}$ be a non-empty set of $\mathcal{K}$, and $J_k\triangleq \phi^r_k(Y_k^n,q^n)$ be the message sent by relay $k\in \mathcal{K}$, and let $\tilde{Q}=q^n$ be the time-sharing variable. For simplicity we define
$X_{\mathcal{L}}^n \triangleq (X_{1}^n,\ldots,X_{L}^n )$, $R_{\mathcal{T}}\triangleq \sum_{t\in \mathcal{T}}R_t$ and $C_{\mathcal{S}}\triangleq \sum_{k\in \mathcal{S}}C_k$.

\noindent Define 
%\begin{align}
$U_{i,k}\triangleq (J_k,Y_k^{i-1}) \quad \text{and }\quad \bar{Q}_i\triangleq (X_{\mathcal{L}}^{i-1},X_{\mathcal{L},i+1}^n,\tilde{Q})$. 
%\end{align}

From Fano's inequality, we have with $\epsilon_n\rightarrow 0$ for $n\rightarrow \infty$ (for vanishing probability of error), for all $\mathcal{T}\subseteq \mathcal{L}$,
\begin{align}
H(m_{\mathcal{T}}|J_{\mathcal{K}},F_{\mathcal{L}},\tilde{Q}) \leq H(m_{\mathcal{L}}|J_{\mathcal{K}},F_{\mathcal{L}},\tilde{Q})\leq n \epsilon_n.\label{eq:MK_C_Fano}
\end{align}
We show the following inequality, used below in the proof.
\begin{align}
H(X^n_{\mathcal{T}}|X^n_{\mathcal{T}^c},J_{\mathcal{K}},\tilde{Q})&\leq \sum_{i=1}^nH(X_{\mathcal{T},i}|X_{\mathcal{T}^c,i},\bar{Q}_i) -nR_{\mathcal{ T}}\\ 
&  \triangleq n\Gamma_{\mathcal{T}}  .\label{eq:MK_C_BoundGamma}
\end{align}
Inequality \eqref{eq:MK_C_BoundGamma} can be shown as follows. 
\begin{align}
%nR_{\mathcal{T}}
n&R_{\mathcal{T}} = H(m_{\mathcal{T}})\label{eq:MK_C_Ineq_00}\\
=& I(m_{\mathcal{T}};J_{\mathcal{K}},F_{\mathcal{L}},\tilde{Q})+ H(m_{\mathcal{T}}|J_{\mathcal{K}},F_{\mathcal{L}},\tilde{Q})\\
=& I(m_{\mathcal{T}};J_{\mathcal{K}},F_{\mathcal{T}}|F_{\mathcal{T}^c},\tilde{Q})+ H(m_{\mathcal{T}}|J_{\mathcal{K}},F_{\mathcal{L}},\tilde{Q})\label{eq:MK_C_Ineq_0}\\
\leq& I(m_{\mathcal{T}};J_{\mathcal{K}},F_{\mathcal{T}}|F_{\mathcal{T}^c},\tilde{Q})+ n\epsilon_n\label{eq:MK_C_Fano_sub}\\
=& H(J_{\mathcal{K}},F_{\mathcal{T}}|F_{\mathcal{T}^c},\tilde{Q}) - H(J_{\mathcal{K}},F_{\mathcal{T}}|F_{\mathcal{T}^c},m_{\mathcal{T}},\tilde{Q})+ n\epsilon_n\\
=& H(J_{\mathcal{K}}|F_{\mathcal{T}^c},\tilde{Q})+H(F_{\mathcal{T}}|F_{\mathcal{T}^c},J_{\mathcal{K}},\tilde{Q}) + n\epsilon_n\\
&- H(F_{\mathcal{T}}|F_{\mathcal{T}^c},m_{\mathcal{T}},\tilde{Q}) - H(J_{\mathcal{K}}|F_{\mathcal{T}^c},m_{\mathcal{T}},F_{\mathcal{T}},\tilde{Q})\\
 =& I(m_{\mathcal{T}},F_{\mathcal{T}};J_{\mathcal{K}}|F_{\mathcal{T}^c},\tilde{Q})-I(F_{\mathcal{T}};J_{\mathcal{K}}|F_{\mathcal{T}^c},\tilde{Q})+ n\epsilon_n \label{eq:MK_C_indep}\\\
\leq &  I(m_{\mathcal{T}},F_{\mathcal{T}};J_{\mathcal{K}}|F_{\mathcal{T}^c},\tilde{Q})+ n\epsilon_n\\
\leq &  I(X_{\mathcal{T}}^n;J_{\mathcal{K}}|F_{\mathcal{T}^c},\tilde{Q})+ n\epsilon_n\label{eq:MK_C_dataproc}\\
 = & H(X_{\mathcal{T}}^n|F_{\mathcal{T}^c},\tilde{Q})-H(X_{\mathcal{T}}^n|F_{\mathcal{T}^c},J_{\mathcal{K}},\tilde{Q})+ n\epsilon_n\label{eq:MK_C_Ineq_10}\\
 % = & H(X_{\mathcal{T}}^n|X_{\mathcal{T}^c}^n,F_{\mathcal{T}^c},\tilde{Q})-H(X_{\mathcal{T}}^n|F_{\mathcal{T}^c},J_{\mathcal{K}},\tilde{Q})+ n\epsilon_n
  \leq & H(X_{\mathcal{T}}^n|X_{\mathcal{T}^c}^n,\tilde{Q})-H(X_{\mathcal{T}}^n|X_{\mathcal{T}^c}^n,F_{\mathcal{T}^c},J_{\mathcal{K}},\tilde{Q})+ n\epsilon_n\label{eq:MK_C_Ineq_11}\\
  = & H(X_{\mathcal{T}}^n|X_{\mathcal{T}^c}^n,\tilde{Q})-H(X_{\mathcal{T}}^n|X_{\mathcal{T}^c}^n,J_{\mathcal{K}},\tilde{Q})+ n\epsilon_n,\label{eq:MK_C_Ineq_12}
\end{align}
where \eqref{eq:MK_C_Ineq_00} follows since $m_{\mathcal{T}}$ are independent; \eqref{eq:MK_C_Ineq_0} follows since $m_{\mathcal{T}}$ is independent of $\tilde{Q}$ and $F_{\mathcal{T}^c}$; \eqref{eq:MK_C_Fano_sub} follows from \eqref{eq:MK_C_Fano}; \eqref{eq:MK_C_indep} follows since $m_{\mathcal{T}}$ is independent of $F_{\mathcal{L}}$; \eqref{eq:MK_C_dataproc} follows from the data processing inequality;\eqref{eq:MK_C_Ineq_11} follows since $X_{\mathcal{T}^c}^n,F_{\mathcal{T}^c}$ are independent from $X_{\mathcal{T}}^n$ and since conditioning reduces entropy and; \eqref{eq:MK_C_Ineq_12} follows  due to the Markov chain 
\begin{align}
X_{\mathcal{T}}^n \mkv (X_{\mathcal{T}^c}^n,J_{\mathcal{K}},\tilde{Q}) \mkv F_{\mathcal{T}^c}.
\end{align}
Then, from \eqref{eq:MK_C_Ineq_12} we have \eqref{eq:MK_C_BoundGamma} as follows: 
\begin{align}
&H(X_{\mathcal{T}}^n|X_{\mathcal{T}^c}^n,J_{\mathcal{K}},\tilde{Q})\\
%&\leq  H(X_{\mathcal{T}}^n|X_{\mathcal{T}^c}^n,\tilde{Q}) -nR_{\mathcal{T}}- n\epsilon_n\\
&\leq \sum_{i=1}^{n}H(X_{\mathcal{T},i}|X_{\mathcal{T}^c}^n, X_{\mathcal{T}}^{i-1},\tilde{Q}) - nR_{\mathcal{T}}\label{eq:MK_C_Ineq_0_1}\\
&= \sum_{i=1}^{n}H(X_{\mathcal{T},i}|X_{\mathcal{T}^c,i},X_{\mathcal{L}}^{i-1},X_{\mathcal{L},i+1}^n,\tilde{Q}) - nR_{\mathcal{T}}\label{eq:MK_C_Ineq_0_2}\\
&= \sum_{i=1}^{n} H(X_{\mathcal{T},i}|X_{\mathcal{T}^c,i},\bar{Q}_i) -nR_{\mathcal{T}} = n\Gamma_{\mathcal{T}}.
\end{align}
where \eqref{eq:MK_C_Ineq_0_2} is due to Lemma \ref{lem:IIDinput}.

Continuing from \eqref{eq:MK_C_Ineq_12}, we have
\begin{align}
 nR_{\mathcal{T}} \leq &
 %  H(X_{\mathcal{T}}^n|X_{\mathcal{T}^c}^n,\tilde{Q})-H(X_{\mathcal{T}}^n|X_{\mathcal{T}^c}^n, J_{\mathcal{K}},\tilde{Q})+ n\epsilon_n\\
%
%=
\sum_{i=1}^n H(X_{\mathcal{T},i}|X_{\mathcal{T}^c}^n,\tilde{Q}, X_{\mathcal{T}}^{i-1})
\nonumber\\
&
-H(X_{\mathcal{T},i}^n|X_{\mathcal{T}^c}^n, J_{\mathcal{K}},X_{\mathcal{T}}^{i-1},\tilde{Q})+ n\epsilon_n\\
=& \sum_{i=1}^n H(X_{\mathcal{T},i}|X_{\mathcal{T}^c}^n, \tilde{Q}, X_{\mathcal{T}}^{i-1}, X_{\mathcal{T},i+1}^{n})\nonumber\\
&-H(X_{\mathcal{T},i}|X_{\mathcal{T}^c}^n, J_{\mathcal{K}},X_{\mathcal{T}}^{i-1},\tilde{Q})+ n\epsilon_n\label{eq:MK_C_ineq_2}\\
\leq&
%\sum_{i=1}^n H(X_{\mathcal{T},i}|X_{\mathcal{T}^c,i},\bar{Q}_i)\label{eq:MK_C_ineq_3}\nonumber\\
%&-H(X_{\mathcal{T},i}|X_{\mathcal{T}^c,i},J_{\mathcal{K}}, Y_{\mathcal{K}}^{i-1},X_{\mathcal{L}}^{i-1},X_{\mathcal{L},i+1}^{n},\tilde{Q})+ n\epsilon_n\\
%
%=&
 \sum_{i=1}^n H(X_{\mathcal{T},i}|X_{\mathcal{T}^c,i},\bar{Q}_i)\nonumber\\
&-H(X_{\mathcal{T},i}|X_{\mathcal{T}^c,i}, U_{\mathcal{K},i},\bar{Q}_{i})+ n\epsilon_n\label{eq:MK_C_ineq_3}\\
=& \sum_{i=1}^n I(X_{\mathcal{T},i};U_{\mathcal{K},i}|X_{\mathcal{T}^c,i},\bar{Q}_{i})+ n\epsilon_n\label{eq:MK_C_ineq_4},
\end{align}
where \eqref{eq:MK_C_ineq_2} follows due to Lemma \ref{lem:IIDinput};  and \eqref{eq:MK_C_ineq_3} follows since conditioning reduces entropy.

On the other hand, we have the following equality
\begin{align}
&I(Y_{\mathcal{S}}^n;J_{\mathcal{S}}|X_{\mathcal{L}}^n, J_{\mathcal{S}^c},\tilde{Q})= \sum_{k\in\mathcal{S}}I(Y_k^n;J_{k}|X_{\mathcal{L}}^n,\tilde{Q})\label{eq:MK_C_Ineq_2_2}\\
&=\sum_{k\in\mathcal{S}}\sum_{i=1}^nI(Y_{k,i};J_k|X_{\mathcal{L}}^n,Y_{k}^{i-1},\tilde{Q})\\
&=\sum_{k\in\mathcal{S}}\sum_{i=1}^nI(Y_{k,i};J_k,Y_{k}^{i-1}|X_{\mathcal{L}}^n,\tilde{Q})\label{eq:MK_C_Ineq_2_4}\\
&=  \sum_{k\in\mathcal{S}}\sum_{i=1}^nI(Y_{k,i};U_{k,i}|X_{\mathcal{L},i},\bar{Q}_i), \label{eq:MK_C_Ineq_2_41}
\end{align}
where \eqref{eq:MK_C_Ineq_2_2} follows due to the Markov chain
\begin{align}
J_k \mkv Y_k^n \mkv X_{\mathcal{L}}^n \mkv Y_{\mathcal{S}\setminus k}^n \mkv J_{\mathcal{S}\setminus k}\quad \text{for } k=[1, K],
\end{align} 
 and since $J_k$ is a function of  of $Y_k^n$; and \eqref{eq:MK_C_Ineq_2_4} follows due to the Markov chain $Y_{k,i}-X_{\mathcal{L}}^n-Y_{k}^{i-1}$ which follows since the channel is memoryless.

Then, from the relay side we have, 
\begin{align}
%n\sum_{k\in \mathcal{S}}C_k
nC_{\mathcal{S}}\geq& \sum_{k\in \mathcal{S}}H(J_k)
\geq  H(J_{\mathcal{S}})\\
\geq & H(J_{\mathcal{S}}|X_{\mathcal{T}^c}^n,J_{\mathcal{S}^c},\tilde{Q})\\
\geq& I(Y_{\mathcal{S}}^n;J_{\mathcal{S}}|X_{\mathcal{T}^c}^n,J_{\mathcal{S}^c},\tilde{Q})\\
=& I(X_{\mathcal{T}}^n,Y_{\mathcal{S}}^n;J_{\mathcal{S}}|X_{\mathcal{T}^c}^n,J_{\mathcal{S}^c},\tilde{Q})\label{eq:MK_C_Ineq_2_1}\\
%=& I(X_{\mathcal{T}}^n;J_{\mathcal{S}}|X_{\mathcal{T}^c}^n,J_{\mathcal{S}^c},\tilde{Q})+I(Y_{\mathcal{S}}^n;J_{\mathcal{S}}|X_{\mathcal{L}}^n, J_{\mathcal{S}^c},\tilde{Q})\\
%
=& H(X_{\mathcal{T}}^n|X_{\mathcal{T}^c}^n,J_{\mathcal{S}^c},\tilde{Q})-H(X_{\mathcal{T}}^n|X_{\mathcal{T}^c}^n, J_{\mathcal{K}},\tilde{Q})\nonumber\\
&+ I(Y_{\mathcal{S}}^n;J_{\mathcal{S}}|X_{\mathcal{L}}^n, J_{\mathcal{S}^c},\tilde{Q})\\
\geq& H(X_{\mathcal{T}}^n|X_{\mathcal{T}^c}^n,J_{\mathcal{S}^c},\tilde{Q})-n\Gamma_{\mathcal{T}}\nonumber\\
&+ I(Y_{\mathcal{S}}^n;J_{\mathcal{S}}|X_{\mathcal{L}}^n, J_{\mathcal{S}^c},\tilde{Q})\label{eq:MK_C_Ineq_2_3}\\
=& \sum_{i=1}^nH(X_{\mathcal{T},i}|X_{\mathcal{T}^c}^n,J_{\mathcal{S}^c},X_{\mathcal{T}}^{i-1},\tilde{Q})-n\Gamma_{\mathcal{T}}\nonumber\\
&+I(Y_{\mathcal{S}}^n;J_{\mathcal{S}}|X_{\mathcal{L}}^n, J_{\mathcal{S}^c},\tilde{Q})\\
\geq & \sum_{i=1}^nH(X_{\mathcal{T},i}|X_{\mathcal{T}^c,i},U_{\mathcal{S}^c,i},\bar{Q}_i)-n\Gamma_{\mathcal{T}}\nonumber\\
&+ I(Y_{\mathcal{S}}^n;J_{\mathcal{S}}|X_{\mathcal{L}}^n, J_{\mathcal{S}^c},\tilde{Q})\label{eq:MK_C_Ineq_2_5}\\
%
%=&\sum_{i=1}^nH(X_{\mathcal{T},i}| X_{\mathcal{T}^c,i},U_{\mathcal{S}^c,i},\bar{Q}_i)- H(X_{\mathcal{T},i}|X_{\mathcal{T}^c,i},\bar{Q}_i)\nonumber\\
%& +nR_{\mathcal{T}}  + \sum_{k\in\mathcal{S}}\sum_{i=1}^nI(Y_{k,i};U_{k,i}|X_{\mathcal{L},i},\bar{Q}_i)\label{eq:MK_C_Ineq_3_0}\\
%
=&nR_{\mathcal{T}} - \sum_{i=1}^nI(X_{\mathcal{T},i};U_{\mathcal{S}^c,i}|X_{\mathcal{T}^c,i},\bar{Q}_i)\nonumber\\ 
&+\sum_{k\in\mathcal{S}}\sum_{i=1}^nI(Y_{k,i};U_{k,i}|X_{\mathcal{L},i},\bar{Q}_i),\label{eq:MK_C_Ineq_3_1}
\end{align}
where~\eqref{eq:MK_C_Ineq_2_1} follows since $J_{\mathcal{S}}$ is a function of $Y_{\mathcal{S}}^n$;  \eqref{eq:MK_C_Ineq_2_3} follows from \eqref{eq:MK_C_BoundGamma};  \eqref{eq:MK_C_Ineq_2_5} follows since conditioning reduces entropy; and \eqref{eq:MK_C_Ineq_3_1} follows from \eqref{eq:MK_C_BoundGamma} and \eqref{eq:MK_C_Ineq_2_41}.

In general,  $\bar{Q}_i$ is not independent of $X_{\mathcal{L},i},Y_{\mathcal{S},i}$, and that due to Lemma \ref{lem:IIDinput}, conditioned on $\bar{Q}_i$, we have the Markov chain 
\begin{align}
U_{k,i}-Y_{k,i}-X_{\mathcal{L},i}-Y_{{\mathcal{K}\setminus k},i}-U_{{\mathcal{K}\setminus k},i}. 
\end{align}

\iffalse
Finally, we define the standard time-sharing variable $Q'$  uniformly distributed over $\{1,\ldots, n\}$, $X_{\mathcal{L}} \triangleq X_{\mathcal{L},Q'}$, $Y_k\triangleq Y_{k,Q'}$, $U_k \triangleq U_{k,Q'}$ and $Q \triangleq [\bar{Q}_{Q'},Q']$ and 
we have from \eqref{eq:MK_C_ineq_4} 
\begin{align}
nR_{\mathcal{T}}&\leq \sum_{i=1}^n I(X_{\mathcal{T},i};U_{\mathcal{K},i}|X_{\mathcal{T},i},\bar{Q}_{i})+ n\epsilon_n\\
&= nI(X_{\mathcal{T},Q'};U_{\mathcal{K},Q'}|X_{\mathcal{T}^c,Q'},\bar{Q}_{Q'},Q')+ n\epsilon_n\\
&= nI(X_{\mathcal{T}};U_{\mathcal{K}}|X_{\mathcal{T}^c},Q)+ n\epsilon_n,
\end{align}
and similarly, from \eqref{eq:MK_C_Ineq_3_1}, we have
\begin{align}
R_{\mathcal{T}}\leq
& C_{\mathcal{S}}  - \sum_{k\in\mathcal{S}}I(Y_{k};U_{k}|X_{\mathcal{L}},Q) + I(X_{\mathcal{L}};U_{\mathcal{S}^c}|X_{\mathcal{T}^c},Q).\nonumber
\end{align}
This completes the proof of Theorem~\ref{th:MK_C_Main}.\qed
\fi

\noindent Finally, a standard time-sharing argument completes the proof of Theorem~\ref{th:MK_C_Main}.\qed

\vspace{-0.3cm}

\section{Concluding Remarks}

In this paper, we study transmission over a cloud radio access network under the framework of oblivious processing at the relay nodes, i.e., the relays are not allowed to know, or cannot acquire, the users' codebooks. Our results shed light (and sometimes determine exactly) what operations the relay nodes should perform optimally in this case. In particular, perhaps non-surprisingly it is shown that compress-and-forward, or variants of it, generally perform well in this case, and are optimal when the outputs at the relay nodes are conditionally independent on the users inputs. Furthermore, in addition to its relevance from a practical viewpoint, restricting the relays not to know/utilize the users' codebooks causes only a bounded rate loss in comparison with the non-oblivious setting (e.g., compress-and-forward and noisy network coding perform to within a constant gap from the cut-set bound in the Gaussian case).  

\vspace{-0.5cm}

\bibliographystyle{ieeetran}
\bibliography{ref}

% APPENDICES TO BE RELEGATED TO A SEPARATE FILE

\end{document}